\pdfoutput=1

\documentclass{lipics-v2019}

\pdfoutput=1

\usepackage[dvipsnames]{xcolor}
\usepackage[backgroundcolor=magenta!10!white]{todonotes}

\usepackage{microtype}
\usepackage{amsmath,amssymb,amsfonts}
\usepackage{thm-restate}

\usepackage{graphicx}

\usepackage{ifthen}

\usepackage{numprint}
\npdecimalsign{.}
\npthousandsep{,}

\usepackage{hyperref}
\usepackage{tikz}
\usetikzlibrary{arrows,automata,positioning,calc,intersections,backgrounds}
\usepackage{xcolor}

\def\<{\langle}
\def\>{\rangle}

\def\Nat{\mathbb{N}}

\def\Real{\mathbb{R}}

\def\cA{\mathcal{A}}

\def\cD{\mathcal{D}}

\def\cL{\mathcal{L}}
\def\cM{\mathcal{M}}

\def\Pr{\mathrm{Pr}}

\newcommand{\CoPath}{\mathit{CoPath}}
\newcommand{\Co}{\mathit{Co}}
\newcommand{\Survives}{\mathit{Survives}}
\newcommand{\unvisited}{\mathit{worklist}}

\colorlet{davidColor}{YellowGreen!30!white}
\colorlet{stefanColor}{Brown!30!white}
\colorlet{benColor}{Yellow!30!white}

\newcommand{\then}{\mathop{\triangleright}}

\newcommand{\cvec}[1]{\pmb{\left[\vphantom{#1}\right.} #1 \pmb{\left.\vphantom{#1}\right]}}

\title
{Efficient Analysis of Unambiguous Automata Using Matrix Semigroup Techniques}

\author{Stefan Kiefer}{University of Oxford, UK}{}{}{Work supported by a Royal Society University Research Fellowship.}
\author{Cas Widdershoven}{University of Oxford, UK}{}{}{}

\authorrunning{Stefan Kiefer and Cas Widdershoven} 

\Copyright{Stefan Kiefer and Cas Widdershoven}



\category{}

\relatedversion{}

\supplement{}

\nolinenumbers 

\begin{document}

\maketitle              

\begin{abstract}
We introduce a novel technique to analyse unambiguous B\"uchi automata quantitatively, and apply this to the model checking problem.
It is based on linear-algebra arguments that originate from the analysis of matrix semigroups with constant spectral radius.
This method can replace a combinatorial procedure that dominates the computational complexity of the existing procedure by Baier et al.
We analyse the complexity in detail, showing that, in terms of the set $Q$ of states of the automaton, the new algorithm runs in time $O(|Q|^4)$, improving on an efficient implementation of the combinatorial algorithm by a factor of $|Q|$.
\end{abstract}

\newpage

\section{Introduction}

Given a finite automaton~$\cA$, what is the proportion of words accepted by it?
This question is natural but imprecise: there are infinitely many words and the proportion of accepted words may depend on the word length.
One may consider the sequence $d_0, d_1, \ldots$ where $d_i$ is the proportion of length-$i$ words accepted by~$\cA$, i.e., $d_i = \frac{|L(\cA) \cap \Sigma^i|}{|\Sigma|^i}$.
The sequence does not necessarily converge, but one may study, e.g., possible limits and accumulation points~\cite{Bodirsky04}.

Alternatively, one can specify a probability distribution on words, e.g., with a Markov chain, and ask for the probability that a word is accepted by~$\cA$.
For instance, if $\Sigma = \{a,b\}$, one may generate a random word, letter by letter, by outputting $a$, $b$ with probability~$1/3$ each, and ending the word with probability~$1/3$.
For an NFA~$\cA$, determining whether the probability of generating an accepted word is~$1$ is equivalent to universality (is $L(\cA) = \Sigma^*$?), a PSPACE-complete problem.
However, if $\cA$ is unambiguous, i.e., every accepted word has exactly one accepting path, then one can compute the probability of generating an accepted word in polynomial time by solving a linear system of equations.
Unambiguousness allows us to express the probability of a union as the sum of probabilities:

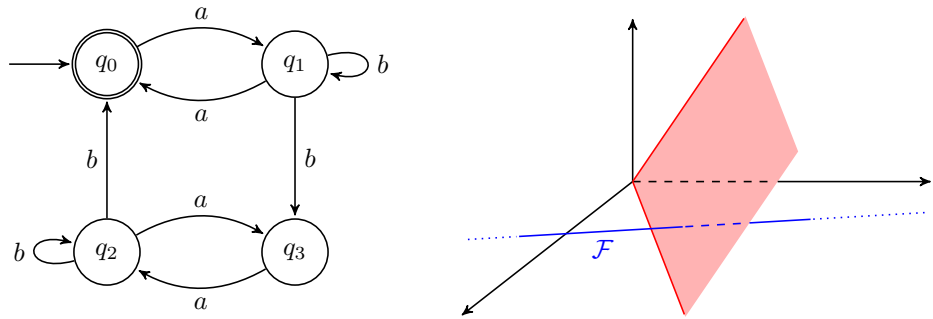
\begin{figure}
\begin{center}
\begin{tikzpicture}[->,>=stealth',shorten >=1pt,auto,node distance=2.5cm,
                    semithick]
  \tikzstyle{every state}=[draw=black,text=black]

  \node[accepting,state] (A)                    {$q_0$};
  \node[state]         (B)       [right of=A] {$q_1$};
  \node[state]         (C)       [below of=A] {$q_2$};
  \node[state]         (D)       [below of=B] {$q_3$};

 \draw[->] (-1.3,0) -- (A);
 \path (A) edge [bend left] node {$a$} (B);
 \path (B) edge [loop right] node {$b$} (B);
 \path (B) edge node {$b$} (D);
 \path (B) edge [bend left] node {$a$} (A);
 \path (C) edge [loop left] node {$b$} (C);
 \path (C) edge node {$b$} (A);
 \path (C) edge [bend left] node {$a$} (D);
 \path (D) edge [bend left] node {$a$} (C);

\end{tikzpicture}
\qquad
\begin{tikzpicture}[>=stealth',shorten >=1pt,auto,semithick]

\coordinate (o) at (0,0);
\coordinate (x) at (4,0);
\coordinate (y) at (0,2.2);
\coordinate (z) at (-2.3,-1.8);
\coordinate (fy)at (0.7,-1.8);
\coordinate (fz)at (1.5,2.2);
\coordinate (f) at ($(fy)+(fz)$);
\coordinate (l3) at (0.7,-0.6);
\coordinate (l5) at (2.4,-0.5);
\coordinate (l6) at ($(l3)!1.9!(l5)$);
\coordinate (l2) at ($(l3)!-1.3!(l5)$);
\coordinate (l1) at ($(l3)!-1.7!(l5)$);

\path [name path=l3l5] (l3) -- (l5);
\fill[red!30] (o) -- (fz) -- (f) -- (fy) -- cycle;
\path [name path=ox] (o) -- (x);
\path [name path=fyf] (fy) -- (f);
\path [name intersections={of=ox and fyf,by={xh}}];
\path [name intersections={of=l3l5 and fyf,by={l4}}];
\draw[dashed]     (o) -- (xh);
\draw[->] (xh) -- (x);
\draw[->] (o) -- (y);
\draw[->] (o) -- (z);
\draw[red](o) -- ($(o)!1!(fy)$);
\draw[red](o) -- (fz);
\draw[blue,dotted] (l1) -- (l2);
\draw[blue]        (l2) -- (l3);
\draw[blue,dashed] (l3) -- (l4);
\draw[blue]        (l4) -- (l5);
\draw[blue,dotted] (l5) -- (l6);
\node[blue] at (-0.4,-0.9) {$\mathcal{F}$};
\end{tikzpicture}
\end{center}
\caption{Left: unambiguous automaton~$\cA$. Right: visualisation of the affine space~$\mathcal{F}$ (blue) and the vector space spanned by (pseudo-)cuts (red); these spaces are orthogonal.}
\label{fig:automaton}
\end{figure}
\begin{example} \label{ex-intro-finite}
Consider the unambiguous automaton~$\cA$ in Figure~\ref{fig:automaton} (left).
If we generate a random word over~$\{a,b\}$ according to the process described above,
we have the following linear system for the vector $\vec{z}$ where $\vec{z}_q$ is, for each $q \in \{q_0, q_1, q_2, q_3\}$, the probability that the word is accepted when $q$ is taken as initial state:
\begin{align*}
 \vec{z}_{q_0} & \ = \ \textstyle\frac13 \vec{z}_{q_1} + \frac13
 & \vec{z}_{q_1}  & \ = \ \textstyle\frac13 \vec{z}_{q_0} + \frac13 (\vec{z}_{q_1} + \vec{z}_{q_3}) \\
 \vec{z}_{q_2} & \ = \ \textstyle\frac13 \vec{z}_{q_3} + \frac13 (\vec{z}_{q_0} + \vec{z}_{q_2})
 & \vec{z}_{q_3} & \ = \ \textstyle\frac13 \vec{z}_{q_2}
\end{align*}
The constant term in the equation for $\vec{z}_{q_0}$ reflects the fact that $q_0$ is accepting.
The (other) coefficients $\frac13$ correspond to the production of either $a$ or~$b$.
The linear system has a unique solution.
\end{example}
One may view an NFA $\cA$ as a B\"uchi automaton, so that its language $L(\cA) \subseteq \Sigma^\omega$ is the set of those infinite words that have an accepting run in~$\cA$, i.e., a run that visits accepting states infinitely often.
There is a natural notion of an infinite random word over~$\Sigma$: in each step sample a letter from~$\Sigma$ uniformly at random, e.g., if $\Sigma=\{a,b\}$ then choose $a$ and~$b$ with probability~$1/2$ each.
Perhaps more significantly, model checking Markov chains against B\"uchi automata, i.e., computing the probability that the random word generated by the Markov chain is accepted by the automaton, is a key problem in the verification of probabilistic systems.
Unfortunately, like the aforementioned problem on finite words, it is also PSPACE-complete~\cite{CY95}.
However, if the B\"uchi automaton is unambiguous, i.e., every accepted (infinite) word has exactly one accepting path, then one can compute the probability of generating an accepted word in polynomial time~\cite{16BKKKMW-CAV}, both in the given B\"uchi automaton and in a given (discrete-time, finite-state) Markov chain.
Since LTL specifications can be converted to unambiguous B\"uchi automata with a single-exponential blow-up, this leads to an LTL model-checking algorithm with single-exponential runtime, which is optimal.
The polynomial-time algorithm from~\cite{16BKKKMW-CAV} for unambiguous B\"uchi automata is more involved than in the finite-word case.
\begin{example} \label{ex-intro-infinite}
In the following we view the automaton~$\cA$ from Figure~\ref{fig:automaton} as an (unambiguous) B\"uchi automaton.
If we generate a random word over~$\{a,b\}$ according to the process described above,
then the vector $\vec{z}$ where for each $q \in \{q_0, q_1, q_2, q_3\}$, $\vec{z}_q$ is the probability that the word is accepted when $q$ is taken as initial state, is a solution to the following linear system:
\begin{align*}
 \vec{z}_{q_0} & \ = \ \textstyle\frac12 \vec{z}_{q_1}
 & \vec{z}_{q_1}  & \ = \ \textstyle\frac12 \vec{z}_{q_0} + \frac12 (\vec{z}_{q_1} + \vec{z}_{q_3}) \\
 \vec{z}_{q_2} & \ = \ \textstyle\frac12 \vec{z}_{q_3} + \frac12 (\vec{z}_{q_0} + \vec{z}_{q_2})
 & \vec{z}_{q_3} & \ = \ \textstyle\frac12 \vec{z}_{q_2}
\end{align*}
However, this linear system has multiple solutions: indeed, any scalar multiple $(1, 2, 2, 1)^\top$ is a solution.
\end{example}
In order to make such a linear system uniquely solvable, one needs to add further equations, and finding these further equations is where the real challenge lies.
Assuming that the state space~$Q$ of~$\cA$ is strongly connected and the Markov chain generates letters uniformly at random as described above, a single additional equation $\vec{\mu}^\top \vec{z} = 1$ suffices (this can be shown with Perron-Frobenius theory: the eigenspace for the dominant eigenvalue of a nonnegative irreducible matrix is one-dimensional).
We call such a vector $\vec{\mu} \in \mathbb{R}^Q$ a \emph{normaliser}.
The aim of this paper is to use a novel, linear-algebra based technique to compute normalisers more efficiently.

The suggestion in~\cite{16BKKKMW-CAV} was to take as normaliser the characteristic vector $\cvec{c} \in \{0,1\}^Q$ of a so-called \emph{cut} $c \subseteq Q$.
To define this, let us write $\delta(q,w)$ for the set of states reachable from a state~$q \in Q$ via the word~$w \in \Sigma^*$.
A \emph{cut} is a set of states of the form $c = \delta(q,w)$ such that $\delta(q,w x) \ne \emptyset$ holds for all $x \in \Sigma^*$.
If a cut does not exist or if $\cA$ does not have accepting states, then we have $\vec{z} = \vec{0}$.
\begin{example} \label{ex-intro-cut}
In the automaton~$\cA$ from Figure~\ref{fig:automaton}, we have a cut $c = \delta(q_0, a b a) = \{q_0, q_2\}$.
Hence its characteristic vector $\vec{\mu} = (1,0,1,0)^\top$ is a normaliser, allowing us to add the equation $\vec{\mu}^\top \vec{z} = \vec{z}_{q_0} + \vec{z}_{q_2} = 1$.
Now the system is uniquely solvable: $\vec{z} = \frac13 (1,2,2,1)^\top$.
The equation $\vec{z}_{q_0} + \vec{z}_{q_2} = 1$ is valid by an ergodicity argument: intuitively, given a finite word that leads to $q_0$ \emph{and}~$q_2$, a random infinite continuation will almost surely enable an accepting run.
For instance, $\vec{z}_{q_0} = \frac13$ is the probability that a random infinite word over $\{a,b\}$ has an odd number of $a$s before the first $b$. (This holds despite the fact that the word $a b b b \ldots$ is not accepted from~$q_0$.)
\end{example}
In Proposition~\ref{prop:findcut} we show that an efficient implementation of the algorithm from~\cite{16BKKKMW-CAV} for computing a cut runs in time $O(|Q|^5)$. 
Our goal is to find a normaliser $\vec{\mu}$ more efficiently.

The general idea is to move from a combinatorial problem, namely computing a set $c \subseteq Q$, to a continuous problem, namely computing a vector $\vec{\mu} \in \mathbb{R}^Q$.
To illustrate this, note that since we can choose as~$\vec{\mu}$ the characteristic vector of an arbitrary cut, we may also choose a convex combination of such vectors, leading to a normaliser~$\vec{\mu}$ with entries other than $0$ or~$1$.

The technical key ideas of this paper draw on the observation that for unambiguous automata with cuts, the transition matrices generate a semigroup of matrices whose spectral radii are all~$1$.
(The spectral radius of a matrix is the largest absolute value of its eigenvalues.)
This observation enables us to adopt techniques that have recently been devised by Protasov and Voynov~\cite{Protasov17} for the analysis of matrix semigroups with constant spectral radius.
To the best of the authors' knowledge, such semigroups have not previously been connected to unambiguous automata.
This transfer is the main contribution of this paper.

To sketch the gist of this technique, for any $a \in \Sigma$ write $M(a) \in \{0,1\}^{Q \times Q}$ for the transition matrix of the unambiguous automaton~$\cA$, define the average matrix $\overline{M} = \frac{1}{|\Sigma|} \sum_{a \in \Sigma} M(a)$, and let $\vec{y} = \overline{M} \vec{y} \in \mathbb{R}^Q$ be an eigenvector with eigenvalue~$1$ (the matrix~$\overline{M}$ has such an eigenvector if $\cA$ has a cut).
Since the matrix semigroup, $\mathcal{S} \subseteq \{0,1\}^{Q \times Q}$, generated by the transition matrices $M(a)$ has constant spectral radius, it follows from~\cite{Protasov17} that one can efficiently compute an affine space $\mathcal{F} \subseteq \mathbb{R}^Q$ with $\vec{y} \in \mathcal{F}$ and $\vec{0} \not\in \mathcal{F}$ such that for any $\vec{v} \in \mathcal{F}$ and any $M \in \mathcal{S}$ we have $M \vec{v} \in \mathcal{F}$.
Using the fact that $\delta(q,w x)$ is a cut (for all $x \in \Sigma^*$) whenever $\delta(q,w)$ is a cut, one can show that all characteristic vectors of cuts have the same scalar product with all $\vec{v} \in \mathcal{F}$, i.e., all characteristic vectors of cuts are in the vector space orthogonal to~$\mathcal{F}$.
Indeed, we choose as normaliser $\vec{\mu}$ a vector that is orthogonal to~$\mathcal{F}$.
This linear-algebra computation can be carried out in time $O(|Q|^3)$.
In the visualisation on the right of Figure~\ref{fig:automaton}, the characteristic vectors of cuts lie in the plane shaded in red, which is orthogonal to straight line~$\mathcal{F}$ (blue).
\begin{example} \label{ex-affine-intro}
In the automaton~$\cA$ from Figure~\ref{fig:automaton}, the vector $\vec{y} = (1,2,2,1)^\top$ satisfies $\overline{M} \vec{y} = \vec{y}$ where $\overline{M} = \frac12 (M(a) + M(b))$.
The affine space $\mathcal{F} := \{\vec{y} + s (1,-1,-1,1)^\top \mid s \in \mathbb{R}\}$ has the mentioned closure properties, i.e., $M(a) \mathcal{F} \subseteq \mathcal{F}$ and $M(b) \mathcal{F} \subseteq \mathcal{F}$.
Note that the vector $\vec{\mu}$ from Example~\ref{ex-intro-cut} is indeed orthogonal to~$\mathcal{F}$, i.e., $\vec{\mu}^\top (1,-1,-1,1)^\top = 0$.
\end{example}
However, to ensure that $\vec{\mu}$ is a valid normaliser, we need to restrict it further.
To this end, we compute, for some state $q \in Q$, the set $\Co(q) \subseteq Q$ of \emph{co-reachable} states, i.e., states $r \in Q$ such that $\delta(q,w) \supseteq \{q,r\}$ holds for some $w \in \Sigma^*$.
This requires a combinatorial algorithm, which is similar to a straightforward algorithm that would verify the unambiguousness of~$\cA$.
Its runtime is quadratic in the number of transitions of~$\cA$, i.e., $O(|Q|^4)$ in the worst case.
Then we restrict~$\vec{\mu}$ such that $\vec{\mu}_q = 1$ and $\mu$ is non-zero only in entries that correspond to~$\Co(q)$.
In the visualisation on the right of Figure~\ref{fig:automaton}, restricting some components of~$\vec{\mu}$ to be~$0$ corresponds to the vectors in the shaded (red) plane that lie on the plane described by $q' = 0$ for all $q' \in Q \setminus Co(q)$. 
\begin{example} \label{ex-Cod-intro}
We have $\Co(q_0) = \{q_0, q_2\}$.
So we restrict $\vec{\mu}$ to be of the form $(1,0,x,0)^\top$.
Together with the equation $\vec{\mu}^\top (1,-1,-1,1)^\top = 0$ this implies $\vec{\mu} = (1,0,1,0)^\top$.
The point is that, although this is the same vector computed via a cut in Example~\ref{ex-intro-cut}, the linear-algebra based computation of~$\vec{\mu}$ is more efficient.
\end{example}
In the rest of the paper we analyse the general case of model checking a given Markov chain against a given unambiguous B\"uchi automaton.
The efficiency gain we aim for with our technique can only be with respect to the automaton, not the Markov chain; nevertheless, we analyse in detail the runtime in terms of the numbers of states and transitions in both the automaton and the Markov chain.
The main results are developed in Section~\ref{sec:uba}.
In Section~\ref{sub:linear-system} we describe the general approach from \cite{16BKKKMW-CAV,baieunp1}.
In Section~\ref{sub:cuts} we analyse the runtime of an efficient implementation of the algorithm from~\cite{16BKKKMW-CAV,baieunp1} for computing a cut.
Our main contribution lies in Section~\ref{sub:pseudo-cuts}, where we develop a new approach for computing a normaliser, based on the mentioned spectral properties of the transition matrices in unambiguous automata.
We close in Section~\ref{sec-discussion} with a discussion. 
The full version of this paper~\cite{kief2019} contains an appendix with proofs.

\section{Preliminaries}
\label{sec:prelim}

We assume the reader to be familiar with basic notions of finite automata over infinite words and Markov
chains, see, e.g., \cite{GraedelThomasWilke02,Kulkarni}.
In the following we provide a brief summary of our notation and a few facts related to linear algebra.

%

\subparagraph*{Finite automata.}
A \emph{B\"uchi automaton}
is a tuple $\cA = (Q,\Sigma,\delta,Q_0,F)$
where $Q$ is the finite set of states, $Q_0 \subseteq Q$ is the set of initial
states, $\Sigma$ is the finite alphabet,
$\delta: Q \times \Sigma \to 2^Q$ is
the transition function, and $F \subseteq Q$ is the set of accepting states.
We extend the transition function to $\delta: Q \times \Sigma^* \to 2^Q$
and to $\delta: 2^Q \times \Sigma^* \to 2^Q$ in the standard way.
For $q \in Q$ we write $\cA_q$ for the automaton obtained from~$\cA$ by making $q$ the only initial state.

Given states $q,r\in Q$ and a finite word
$w = a_0 a_1 \cdots a_{n-1} \in \Sigma^*$,
a \emph{run} for $w$ from $q$ to $r$ is a sequence
$q_0 q_1 \cdots q_n \in Q^{n+1}$ with $q_0=q$, $q_n=r$ and
$q_{i+1} \in \delta(q_i,a_i)$ for $i \in \{0, \ldots, n-1\}$.
A \emph{run} in $\cA$ for an infinite word
$w = a_0 a_1 a_2 \cdots \in \Sigma^{\omega}$ is an infinite sequence
$\rho = q_0 q_1 \cdots \in Q^{\omega}$ such that $q_0 \in Q_0$ and
$q_{i+1} \in \delta(q_i,a_i)$ for all $i \in \Nat$.  Run $\rho$ is
called \emph{accepting} if $\inf(\rho) \cap F \ne \emptyset$ where
$\inf(\rho) \subseteq Q$ is the set of states that occur infinitely
often in~$\rho$.  The \emph{language} $\cL(\cA)$ of accepted words
consists of all infinite words $w\in \Sigma^{\omega}$ that have at
least one accepting run.
%
%
%
%
$\cA$~is called  
\emph{unambiguous} if each word $w\in \Sigma^{\omega}$ has at most one accepting run.
We use the acronym UBA for unambiguous B\"uchi automaton.


We define $|\delta| := |\{(q,r) \mid \exists\, a \in \Sigma : r \in \delta(q,a)\}|$, i.e., $|\delta| \le |Q|^2$ is the number of transitions in~$\cA$ when allowing for multiple labels per transition.
In Appendix~\ref{app-A} we give an example that shows that the number of transitions can be quadratic in~$|Q|$, even for UBAs with a strongly connected state space. 
We assume $|Q| \le |\delta|$, as states without outgoing transitions can be removed.
In this paper, $\Sigma$ may be a large set (of states in a Markov chain), so it is imperative to allow for multiple labels per transition.
We use a lookup table to check in constant time whether $r \in \delta(q,a)$ holds for given $r, q, a$.

A \emph{diamond} is given by two states $q, r \in Q$ and a finite word~$w$ such that there exist at least two distinct runs for~$w$ from $q$ to~$r$.
One can remove diamonds (see Appendix~\ref{app:prelims}):
\begin{restatable}{ourlemma}{lemdiamondfree}\label{lem:diamondfree}
Given a UBA, one can compute in time $O(|\delta|^2|\Sigma|)$ a UBA of at most the same size, with the same language and without diamonds.
\end{restatable}
%
\noindent For the rest of the paper, we assume that UBAs do not have diamonds.

%
%

\subparagraph*{Vectors and matrices.}
We consider vectors and square matrices indexed by a finite set~$S$.
We write (column) vectors $\vec{v} \in \Real^S$ with arrows on top,
and $\vec{v}^\top$ for the transpose (a row vector)
of~$\vec{v}$.  The zero vector and the all-ones vector are denoted
by~$\vec{0}$ and $\vec{1}$, respectively.
For a set $T \subseteq S$ we write $\cvec{T} \in \{0,1\}^S$ for the characteristic vector of~$T$, i.e., $\cvec{T}_s = 1$ if $s \in T$ and $\cvec{T}_s = 0$ otherwise.
A matrix
$M \in [0,1]^{S \times S}$ is called \emph{stochastic} if
$M \vec{1} = \vec{1}$, i.e., if every row of $M$ sums to one.  For a
set $U \subseteq S$ we write $\vec{v}_U \in \Real^U$ for the
restriction of~$\vec{v}$ to~$U$.  Similarly, for $T, U \subseteq S$ we
write $M_{T,U}$ for the submatrix of~$M$ obtained by deleting the rows
not indexed by~$T$ and the columns not indexed by~$U$.  The (directed)
\emph{graph} of a nonnegative matrix $M \in \Real^{S \times S}$ has
vertices $s \in S$ and edges $(s,t)$ if $M_{s,t} > 0$.  We may
implicitly associate~$M$ with its graph and speak about
graph-theoretic concepts such as reachability and strongly connected
components (SCCs) in~$M$.




\subparagraph*{Markov chains.}
A (finite-state discrete-time) \emph{Markov chain} is a pair
$\cM = (S,M)$ where $S$ is the finite set of states, and
$M \in [0,1]^{S \times S}$ is a stochastic matrix that specifies
transition probabilities.  An \emph{initial distribution} is a
function $\iota : S \to [0,1]$ satisfying
$\sum_{s \in S} \iota(s) = 1$.  Such a distribution induces a
probability measure~$\Pr^{\cM}_\iota$ on the measurable subsets
of~$S^\omega$ in the standard way, see for instance~\cite[chapter 10.1, page 758]{baie2008}.
If $\iota$ is concentrated on a
single state~$s$, we may write $\Pr^{\cM}_s$
for~$\Pr^{\cM}_\iota$.
We write~$E$ for the set of edges in the graph of~$M$.
Note that $|S| \le |E| \le |S|^2$, as $M$ is stochastic.

\subparagraph*{Solving linear systems.}
Let $\kappa \in [2,3]$ be such that one can multiply two $n \times n$-matrices in time $O(n^\kappa)$ (in other literature, $\kappa$ is often denoted by~$\omega$).
We assume that arithmetic operations cost constant time.
One can choose $\kappa = 2.4$, see \cite{gall2014} for a recent result.
One can check whether an $n \times n$ matrix is invertible in time $O(n^\kappa)$~\cite{bunc1974}.
Finally, one can solve a linear system with $n$ equations using the Moore-Penrose pseudo-inverse~\cite{jame1978} in time $O(n^\kappa)$~\cite{petk2009}.

\subparagraph*{Spectral theory.}
The \emph{spectral radius} of a matrix $M \in \Real^{S \times S}$,
denoted $\rho(M)$, is the largest absolute value of the eigenvalues
of~$M$.
By the Perron-Frobenius theorem~\cite[Theorems 2.1.1, 2.1.4]{book:bermanP94}, if $M$ is nonnegative then the spectral radius $\rho(M)$ is an eigenvalue of $M$ and there is a nonnegative eigenvector $\vec{x}$ with $M\vec{x}=\rho(M)\vec{x}$.
Such a vector~$\vec{x}$ is called \emph{dominant}.
Further, if $M$ is nonnegative and strongly connected then $\vec{x}$ is strictly positive in all components and the eigenspace associated with $\rho(M)$ is one-dimensional.

\section{Algorithms}
\label{sec:uba}

Given a Markov chain~$\cM$, an initial distribution~$\iota$, and a B\"uchi
automaton~$\cA$ whose alphabet is the state space of~$\cM$, the
\emph{probabilistic model-checking problem} is to compute
$\Pr^{\cM}_\iota(\cL(\cA))$.
This problem is PSPACE-complete~\cite{CY95,BusRubVar04}, but solvable in polynomial time if $\cA$ is deterministic.
For UBAs a polynomial-time algorithm was described in~\cite{16BKKKMW-CAV,baieunp1}.
In this paper we obtain a faster algorithm (recall that $E$ is the set of transitions in the Markov chain):
\begin{restatable}{theorem}{thmPMCMCUBA}
  \label{thm:PMC-MC-UBA}
  Given a Markov chain $\cM = (S,M)$, an initial distribution~$\iota$, and
  a UBA~$\cA = (Q,S,\delta,Q_0,F)$,
  one can compute $\Pr^{\cM}_\iota(\cL(\cA))$ in time
  $O(|Q|^\kappa |S|^\kappa + |Q|^3|E|+|\delta|^2|E|)$.
\end{restatable}
\noindent Before we prove this theorem in Section~\ref{sub:pseudo-cuts}, we describe the algorithm from~\cite{16BKKKMW-CAV,baieunp1} and analyse the runtime of an efficient implementation.

\subsection{The Basic Linear System} \label{sub:linear-system}

Let $\cM = (S,M)$ be a Markov chain, $\iota$ an initial distribution.
Let $B \in \Real^{(Q\times S)\times(Q\times S)}$ be the following
matrix:
\begin{equation}\label{eq:Bmat}
B_{\langle q,s\rangle,\langle q',s'\rangle} = \left\{\begin{array}{ll}
M_{s,s'} & \textrm{if $q' \in \delta(q,s)$} \\
0 & \textrm{otherwise}
\end{array}\right.
\end{equation}
Define $\vec{z} \in \Real^{Q\times S}$ by $\vec{z}_{\langle q,s \rangle} = \Pr^{\cM}_s(\cL(\cA_q))$.
Then $\Pr^{\cM}_\iota(\cL(\cA)) = \sum_{q\in Q_0}\sum_{s \in S} \iota(s) \vec{z}_{\langle q,s \rangle}$.
Lemma 4 in \cite{baieunp1} implies that $\vec{z} = B\vec{z}$.

\begin{example}
Consider the UBA~$\cA$ from Figure~\ref{fig:automaton} and the two-state Markov chain~$\cM$ shown on the left of Figure~\ref{fig:product}.
The weighted graph on the right of Figure~\ref{fig:product} represents the matrix~$B$, obtained from
$\cA$ and $\cM$ according to Equation \eqref{eq:Bmat}.  It is natural
to think of $B$ as a product of $\cA$ and $\cM$.  Notice that $B$ is
not stochastic: the sum of the entries in each row (equivalently, the
total outgoing transition weight of a graph node) is not always one.
\label{ex:product}
\end{example}
\begin{figure}
\begin{center}
\begin{tikzpicture}[->,>=stealth',shorten >=1pt,auto,node distance=2.5cm,
                    semithick]
  \tikzstyle{every state}=[draw=black,text=black]
  
  
  \node[accepting,state] (A)                    {$q_0$};
  \node[state]         (B)       [right of=A] {$q_1$};
  \node[state]         (C)       [below of=A] {$q_2$};
  \node[state]         (D)       [below of=B] {$q_3$};

 \draw[->] (-1.3,0) -- (A);
 \path (A) edge [bend left] node {$a$} (B);
 \path (B) edge [loop right] node {$b$} (B);
 \path (B) edge node {$b$} (D);
 \path (B) edge [bend left] node {$a$} (A);
 \path (C) edge [loop left] node {$b$} (C);
 \path (C) edge node {$b$} (A);
 \path (C) edge [bend left] node {$a$} (D);
 \path (D) edge [bend left] node {$a$} (C);
  

\node[state] (X) [below of=C] {$a$};
\node[state] (Y) [right of=X] {$b$};

\path (X) edge [bend left] node {$\frac{1}{2}$} (Y);
\path (Y) edge [bend left] node {$\frac{1}{2}$} (X);
\path (X) edge [loop above] node {$\frac{1}{2}$} (X);
\path (Y) edge [loop above] node {$\frac{1}{2}$} (Y);
\end{tikzpicture}
\quad
\begin{tikzpicture}[->,>=stealth',shorten >=1pt,auto,node distance=2.1cm,
                    semithick, inner sep=1]
  \tikzstyle{every state}=[draw=black,text=black]

  \node[state,inner sep=0pt,scale=0.98]         (A)                    {$\langle q_3,b \rangle$};
  \node[state,inner sep=0pt,scale=0.98]         (B)       [right of=A] {$\langle q_1,b \rangle$};
  \node[state,inner sep=0pt,scale=0.98]	(C)	 [right of=B] {$\langle q_0,a \rangle$};
  \node[state,inner sep=0pt,scale=0.98]         (D)       [above of=C] {$\langle q_1,a \rangle$};
  \node[state,inner sep=0pt,scale=0.98]         (E)       [below of=B] {$\langle q_3,a \rangle$};
  \node[state,inner sep=0pt,scale=0.98]       (F)       [right of=E] {$\langle q_2,b \rangle$};
  \node[state,inner sep=0pt,scale=0.98]	(G)	 [right of=F] {$\langle q_0,b \rangle$};
  \node[state,inner sep=0pt,scale=0.98]       (H)       [below of=E] {$\langle q_2,a \rangle$};

 \path (B) edge node{$\frac{1}{2}$} (A);
 \path (B) edge [loop above] node{$\frac{1}{2}$} (B);
 \path (B) edge node{$\frac{1}{2}$} (D);
 \path (B) edge node{$\frac{1}{2}$} (E);

 \path (C) edge node{$\frac{1}{2}$} (B);
 \path (C) edge [bend left=15] node[inner sep=1pt]{$\frac{1}{2}$} (D);

 \path (D) edge [bend left=15,'] node[inner sep=1pt]{$\frac{1}{2}$} (C);
 \path (D) edge node{$\frac{1}{2}$} (G);

 \path (E) edge node{$\frac{1}{2}$} (F);
 \path (E) edge [bend left=15] node[inner sep=1pt]{$\frac{1}{2}$} (H);

 \path (F) edge node{$\frac{1}{2}$} (C);
 \path (F) edge [loop below] node{$\frac{1}{2}$} (F);
 \path (F) edge node{$\frac{1}{2}$} (G);
 \path (F) edge node{$\frac{1}{2}$} (H);

 \path (H) edge node{$\frac{1}{2}$} (A);
 \path (H) edge [bend left=15,'] node[inner sep=1pt]{$\frac{1}{2}$} (E);
 
  \begin{pgfonlayer}{background}
    \filldraw [line width=4mm,join=round,green!15]
      (D.north  -| H.west)  rectangle (H.south  -| D.east);
    \filldraw [line width=4mm,join=round,red!15]
      (A.north  -| A.west)  rectangle (A.south  -| A.east)
      (G.north  -| G.west)  rectangle (G.south  -| G.east);
    
  \end{pgfonlayer}
\end{tikzpicture}
\end{center}
\caption{The UBA from Figure~\ref{fig:automaton} and the Markov chain $\cM$ on the left, and their product, $B$, on the right. The (single) accepting recurrent SCC is shaded green, and the two other SCCs are shaded red.}
\label{fig:product}
\end{figure}
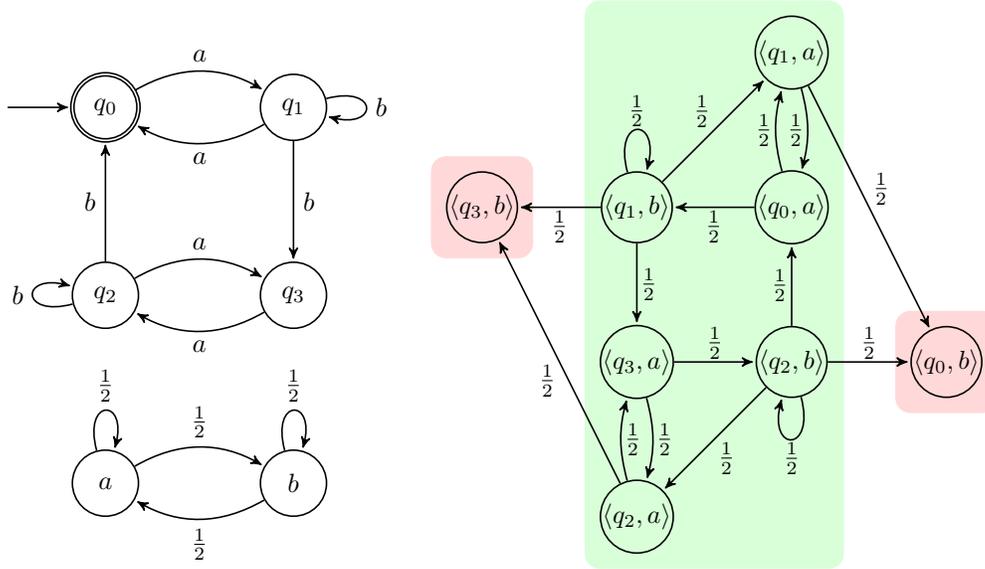
Although $\vec{z}$ is a solution the system of equations $\vec{\zeta} = B\vec{\zeta}$, this system does not uniquely identify $\vec{z}$.
Indeed, any scalar multiple of $\vec{z}$ is a solution for these equations.
To uniquely identify $\vec{z}$ by a system of linear equations, we need to analyse the SCCs of $B$.

All SCCs~$D$ satisfy $\rho(D) \leq 1$, see~\cite[Proposition~7]{baieunp1}.
An SCC $D$ of~$B$ is called \emph{recurrent} if $\rho(B_{D,D})=1$.
It is called \emph{accepting} if there is $\langle q,s \rangle \in D$ with $q \in F$.
%
\begin{example}
The matrix~$B$ from Figure \ref{fig:product} has three SCCs, namely the two singleton sets $\{\langle q_0,b \rangle\}$ and $\{\langle q_3,b \rangle\}$, and $D = \{\langle q_0,a \rangle,\langle q_1,a \rangle,\langle q_1,b \rangle,\langle q_2,a \rangle,\langle q_2,b \rangle,\langle q_3,a \rangle\}$. Only $D$ is recurrent;
indeed, $\vec{y} = (\vec{y}_{\langle q_0,a \rangle},\vec{y}_{\langle q_1,a \rangle},\vec{y}_{\langle q_1,b \rangle},\vec{y}_{\langle q_2,a \rangle},\vec{y}_{\langle q_2,b \rangle},\vec{y}_{\langle q_3,a\rangle})^\top = (2,1,3,1,3,2)^\top$ is a dominant eigenvector with $B_{D,D}\vec{y}=\vec{y}$.
Since $q_0$ is accepting, $D$ is accepting recurrent.
\end{example}
Denote the set of accepting recurrent SCCs by $\cD_+$ and the set of non-accepting recurrent SCCs by $\cD_0$.
By~\cite[Lemma~8]{baieunp1}, for $D \in \cD_+$ we have $\vec{z}_d > 0$ for all $d \in D$, and
for $D \in \cD_0$ we have $\vec{z}_D = \vec{0}$.
Hence, for $D \in \cD_+$, there exists a \emph{$D$-normaliser}, i.e., a vector $\vec{\mu} \in \mathbb{R}^D$ such that $\vec{\mu}^\top \vec{z}_D = 1$.
This gives us a system of linear equations that identifies $\vec{z}$ uniquely~\cite{baieunp1}:
\begin{lemma}[Lemma 12 in \cite{baieunp1}]\label{lem:linear-system}
Let $\cD_+$ be the set of accepting recurrent SCCs, and $\cD_0$ the
set of non-accepting recurrent SCCs.  For each $D \in \cD_+$ let
$\vec{\mu}(D)$ be a $D$-normaliser.  Then $\vec{z}$~is the unique
solution of the following linear system:
\begin{equation}
\begin{aligned}
&&\vec{\zeta} & = B \vec{\zeta} \\
\text{for all } D \in \cD_+ :&& \quad \vec{\mu}(D)^\top \vec{\zeta}_D & = 1 \\
\text{for all } D \in \cD_0 :&& \quad \vec{\zeta}_D & = \vec{0}
\end{aligned}
\label{eq:linear-system}
\end{equation}
\end{lemma}
Uniqueness follows from the fact that the system $\vec{\zeta} = B \vec{\zeta}$ describes the eigenspace of the dominant eigenvalue (here,~$1$) of a nonnegative strongly connected matrix (here,~$B$), and such eigenspaces are one-dimensional.
This leads to the following result:
\begin{proposition}\label{prop:mccomp}
Suppose $N$ is the runtime of an algorithm to calculate a normaliser for each accepting recurrent SCC.
Then one can compute $\Pr^{\cM}_\iota(\cL(\cA))$ in time $O(|Q|^\kappa|S|^\kappa)+N$.
\end{proposition}
\begin{proof}
Lemma~\ref{lem:linear-system} implies correctness of the following procedure to calculate $\Pr^{\cM}_\iota(\cL(\cA))$:
\begin{enumerate}
\item Set up the matrix~$B$ from Equation~\eqref{eq:Bmat}.
\item Compute the SCCs of~$B$.
\item For each SCC $C$, check whether $C$ is recurrent. 
\item For each accepting recurrent SCC~$D$, compute its $D$-normaliser~$\vec{\mu}(D)$.
\item Compute~$\vec{z}$ by solving the linear system~\eqref{eq:linear-system} in Lemma~\ref{lem:linear-system}.
\item Compute $\Pr^{\cM}_\iota(\cL(\cA)) = \sum_{s \in S}\sum_{q \in Q_0}\iota(s)\vec{z}_{q,s}$.
\end{enumerate}
%
%
One can set up $B$ in time $O(|Q|^2|S|^2)$.
Using Tarjan's algorithm one can compute the SCCs of~$B$ in time linear in the vertices and edges of~$B$, hence in $O(|Q|^2|S|^2)$~\cite{tarj1972}. 
One can find those SCCs $D$ which are recurrent in time $O(|Q|^\kappa|S|^\kappa)$ by checking if $I - B_{D,D}$ is invertible.
The linear system~\eqref{eq:linear-system} has $O(|Q||S|)$ equations, and thus can be solved in time $O(|Q|^\kappa|S|^\kappa)$.
Hence the total runtime is $O(|Q|^\kappa|S|^\kappa)+N$. 
\end{proof}
In Section~\ref{sub:cuts} we describe the combinatorial, \emph{cut} based, approach from~\cite{16BKKKMW-CAV,baieunp1} to calculating $D$-normalisers and analyse its complexity.
In Section~\ref{sub:pseudo-cuts} we describe a novel linear-algebra based approach, which is faster in terms of the automaton.

\subsection{Calculating \texorpdfstring{$D$}{D}-Normalisers Using Cuts}
\label{sub:cuts}

For the remainder of the paper, let $D$ be an accepting recurrent SCC.
A \emph{fibre over $s \in S$} is a subset of $D$ of the form $\alpha\times\{s\}$ for some $\alpha \subseteq Q$.
Given a fibre $f = \alpha\times\{s\}$ and a state $s' \in S$, if $M_{s,s'} > 0$ we define the fibre $f \then s'$ as follows:
\begin{equation*}
f \then s' := \{\langle q,s' \rangle \; | \; q \in \delta(\alpha,s)\}\cap D.
\end{equation*}
If $M_{s,s'} = 0$, then $f \then s'$ is undefined, and for $w \in S^*$ we define $f \then w = f$ if $w = \varepsilon$ and $f \then ws' = (f \then w)\then s'$. If $f = \{d\}$ for some $d \in D$ we may write $d \then s'$ for $f \then s'$.

We call a fibre $c$ a \emph{cut} if $c = d \then v$ for some $v \in S^*$ and $d \in D$, and $c \then w \neq \emptyset$ for all $w \in S^*$ whenever $c \then w$ is defined.
Note that if $c$ is a cut then so is $c \then w$ whenever it is defined.
Given a cut $c \subseteq D$ we call its characteristic vector $\cvec{c} \in \{0,1\}^D$ a \emph{cut vector}. In the example in Figure~\ref{fig:product}, it is easy to see that $\{\langle q_1,b\rangle\} = \langle q_0,a\rangle \then b$ is a cut.
%
\begin{lemma}[Lemma 10 in \cite{baieunp1}]\label{lem:cutnormalises}
There exists a cut.
Any cut vector $\vec{\mu}$ is a normaliser, i.e., $\vec{\mu}^\top \vec{z}_D = 1$.
\end{lemma}
Loosely speaking, $\vec{\mu}^\top \vec{z}_D \le 1$ follows from unambiguousness, and $\vec{\mu}^\top \vec{z}_D \not< 1$ follows from an ergodicity argument (intuitively, all states in the cut are almost surely visited infinitely often).
%
%
The following lemma is the basis for the cut computation algorithm in~\cite{16BKKKMW-CAV,baieunp1}:
\begin{lemma}[Lemma 17 in \cite{baieunp1}] \label{lem:get-larger}
Let $D \subseteq Q \times S$ be a recurrent SCC.
Let $d \in D$.
Suppose $w \in S^*$ is such that $d \then w \ni d$ is not a cut.
Then there are $v \in S^*$ and $e \ne d$ with $d \then v \supseteq \{d,e\}$ and $e \then w \ne \emptyset$. For any such $e$, $d \then w \cap e \then w = \emptyset$.
Hence $d \then v w \supseteq \{d,e\} \then w \supsetneq d \then w$.
\end{lemma}
This suggests a way of generating an increasing sequence of fibres, culminating in a cut.
We prove the following proposition:
\begin{restatable}{ourproposition}{propfindcut}\label{prop:findcut}
Let $D \subseteq Q \times S$ be a recurrent SCC.
Denote by $T$ the set of edges in~$B_{D,D}$.
One can compute a cut in time
$O(|Q|^2 |\delta| |D| + |\delta| |T|)$.
\end{restatable}
\noindent Define, for some $d = \langle q,s\rangle \in D$, its \emph{co-reachability} set $\Co(d) \subseteq D$:
it consists of those $e \in D$ such that there exists a word $w$ with $\{d,e\} \subseteq d \then w$.
Note that $\Co(d)$ is a fibre over~$s$. In the example of Figure~\ref{fig:product} we have that $\Co(\langle q_0, a\rangle) = \{\langle q_0, a\rangle,\langle q_2, a\rangle\}$, with $\{\langle q_0, a\rangle, \langle q_2, a\rangle\} \in \langle q_0, a\rangle \then ba$.
The following lemma (proof in Appendix~\ref{app:cuts}) gives a bound on the time to compute~$\Co(d)$:
\begin{restatable}{ourlemma}{lemcalccod}\label{lem:calccod}
One can compute $\Co(d)$ in time $O(|Q| |D| + |\delta| |T|)$.
Moreover, one can compute in time $O(|Q|^2 |D| + |\delta| |T|)$ a list $(\CoPath(d)(e))_{e \in \Co(d)}$ such that $\CoPath(d)(e) \in S^*$ and $\{d,e\} \subseteq d \then \CoPath(d)(e)$ and $|\CoPath(d)(e)| \le |Q| |D|$.
\end{restatable}
\noindent The lemma is used in the proof of Proposition~\ref{prop:findcut}:
\begin{proof}[Proof sketch of Proposition~\ref{prop:findcut}]
Starting from a singleton fibre $\{d\}$, where $d = \langle q,s\rangle \in D$ is chosen arbitrarily, we keep looking for words $v \in S^*$ that have the properties described in Lemma \ref{lem:get-larger} to generate larger fibres $d\then w$:
\begin{enumerate}
\item $w := \varepsilon$ (the empty word)
\item while $\exists\, v \in S^*$ and $\exists\, e \ne d$ such that $d \then v \supseteq \{d, e\}$ and $e \then w \neq \emptyset:$ \\
\mbox{}\hspace{4mm} $w := v w$
\item return $d \then w$.
\end{enumerate}
By \cite[Lemma~18]{16BKKKMW-CAV} the algorithm returns a cut.
In every loop iteration the fibre $d \then w$ increases, so the loop terminates after at most $|Q|$ iterations.
For efficiency we calculate $\Co(d)$ and $\CoPath(d)$ using Lemma~\ref{lem:calccod}, and we use dynamic programming to maintain the set, $\Survives$, of those $e \in D$ for which $e \then w \ne \emptyset$ holds.
Whenever a prefix~$v$ is added to~$w$, we update~$\Survives$ by processing~$v$ backwards.
This leads to the following algorithm:
\begin{enumerate}
\item Calculate $\Co(d)$ and $\CoPath(d)$ using Lemma \ref{lem:calccod}
\item $w := \varepsilon;$ $\Survives := (Q \times \{s\}) \cap D$
\item while $\exists\, e \in \Co(d) \setminus \{d\}$ such that $e \in \Survives$: \\
\mbox{}\hspace{4mm} $v_0= s$; $v_1\ldots v_n := \CoPath(d)(e)$ \\
\mbox{}\hspace{4mm} for $i = n, n-1, \ldots, 1$: \\
\mbox{}\hspace{10mm}     $\Survives := \{\langle p,v_{i-1}\rangle \in D \mid (\delta(p, v_{i-1}) \times \{v_i\}) \cap \Survives \neq \emptyset\}$ \\
\mbox{}\hspace{4mm} $w := v_1\ldots v_nw$
\item return $d \then w$
\end{enumerate}
The runtime analysis is in Appendix~\ref{app:cuts}.
\end{proof}

\begin{example}\label{ex:cut}
Letting $d = \langle q_0,a \rangle$ and $e = \langle q_2,a \rangle$ we have $\Co(d) = \{d, e\}$ with $\CoPath(d)(d) = \varepsilon$ and $\CoPath(d)(e) = b a a$.
Initially we have $\Survives = Q \times \{a\}$.
In the first iteration the algorithm can only pick~$e$.
The inner loop updates $\Survives$ first to $\{q_0, q_1, q_2, q_3\} \times \{a\}$ (i.e., to itself), then to $\{q_1, q_2\} \times \{b\}$, and finally to $\{q_0,q_3\} \times \{a\}$.
Now $(\Co(d) \setminus d) \cap \Survives$ is empty and the loop terminates.
The algorithm returns the cut $d \then baa = \{d,e\}$.
\end{example}
Applying Proposition~\ref{prop:findcut} to the general procedure (Proposition~\ref{prop:mccomp}) leads to the following result on the combinatorial approach:
\begin{restatable}{ourtheorem}{thmPMCMCUBAcut}
  \label{thm:PMC-MC-UBA-cut}
  Given a Markov chain $\cM = (S,M)$, an initial distribution~$\iota$, and
  a UBA~$\cA = (Q,S,\delta,Q_0,F)$, one can compute $\Pr^{\cM}_\iota(\cL(\cA))$
  in time
$O(|Q|^\kappa |S|^\kappa + |Q|^3 |\delta| |S|+|\delta|^2 |E|)$.
\end{restatable}

\subsection{Calculating \texorpdfstring{$D$}{D}-Normalisers Using Linear Algebra}
\label{sub:pseudo-cuts}

Recall that $D$ is an accepting recurrent SCC.
For $t \in S$ define the matrix $\Delta(t) \in \{0,1\}^{D \times D}$ as follows:
\begin{equation*}
\Delta(t)_{\langle q, s\rangle,\langle q', s'\rangle} := \left\{\begin{array}{ll}
1 & \quad\textrm{if $s' = t$, $M_{s,t} > 0$, and $q' \in \delta(q,s)$} \\
0 & \quad\textrm{otherwise}
\end{array}\right.
\end{equation*}
Note that the graph of~$\Delta(t)$ contains exactly the edges of the graph of~$B_{D,D}$ that end in vertices in $Q \times \{t\}$.
If $M_{s,t} > 0$ holds for all pairs $(s,t)$, then the matrices $(\Delta(t))_{t \in S}$ generate a semigroup of matrices, all of which have spectral radius~$1$.
Such semigroups were recently studied by Protasov and Voynov~\cite{Protasov17}.
Specifically, Theorem~5 in~\cite{Protasov17} shows that there exists an affine subspace $\mathcal{F}$ of $\mathbb{R}^{D}$ which excludes $\vec{0}$ and is invariant under multiplication by matrices from the semigroup.
Moreover, they provide a way to compute this affine subspace efficiently.
One can show that cut vectors are orthogonal to~$\mathcal{F}$.
The key idea of our contribution is to generalise cut vectors to \emph{pseudo-cuts}, which are vectors $\vec{\mu} \in \mathbb{R}^D$ that are orthogonal to~$\mathcal{F}$.
We will show (in Lemma~\ref{lem:colocal-pseudo-cut} below) how to derive a $D$-normaliser based on a pseudo-cut that is non-zero only in components that are in a co-reachability set~$\Co(d)$ (from Lemma~\ref{lem:calccod}).

If $M_{s,t} = 0$ holds for some $s,t$ (which will often be the case in model checking), then $\Delta(s) \Delta(t)$ is the zero matrix, which has spectral radius~$0$, not~$1$.
Therefore, the results of~\cite{Protasov17} are not directly applicable and we have to move away from matrix semigroups.
In the following we re-develop and generalise parts of the theory of~\cite{Protasov17} so that the paper is self-contained and products of $\Delta(s) \Delta(t)$ with $M_{s,t} = 0$ are not considered.

Let $w = s_1 s_2 \ldots s_n \in S^*$.
Define $\Delta(w) = \Delta(s_1) \Delta(s_2) \cdots \Delta(s_n)$.
We say $w$ is \emph{enabled} if $M_{s_i,s_{i+1}} > 0$ holds for all $i \in \{1,\ldots, n-1\}$.
If $f \subseteq D$ is a fibre over~$s$ such that $s w$ is enabled, we have $\cvec{f \then w}^\top = \cvec{f}^\top\Delta(w)$.
We overload the term \emph{fibre over $s$} to describe any vector $\vec{\mu} \in \mathbb{R}^D$ such that $\vec{\mu}_{\langle q,s'\rangle} = 0$ whenever $s' \neq s$.
We define \emph{pseudo-cuts over~$s$} to be fibres $\vec{\mu}$ over~$s$ such that $\vec{\mu}^\top\Delta(w)\vec{z} = \vec{\mu}^\top\vec{z}$ holds for all $w \in S^*$ such that $s w$ is enabled.
Let $c \subseteq Q \times \{s\}$ be a cut with $s w$ enabled.
Then $c \then w$ is a cut, and $\cvec{c}^\top \Delta(w) \vec{z} = 1 = \cvec{c}^\top \vec{z}$ holds by Lemma~\ref{lem:cutnormalises}.
It follows that cut vectors are pseudo-cuts.
\begin{example}\label{ex:cutvects}
Since $c = \{\langle q_0,a \rangle, \langle q_2,a \rangle\}$ from Example~\ref{ex:cut} is a cut, $\cvec{c}$ is a pseudo-cut over~$a$.
Pseudo-cuts do not need to be combinations of cut vectors:
although the fibre $f = \{\langle q_0,a\rangle, \langle q_1,a\rangle\}$ is not a cut, $\cvec{f}$ is a pseudo-cut over~$a$.
\end{example}

Fix some $d = \langle q,s \rangle \in D$.
Recall that $\Co(d)$ consists of those $e \in D$ such that there exists a word $w$ with $\{d,e\} \subseteq d \then w$.
We define \emph{$\Co(d)$-pseudo-cuts} to be pseudo-cuts $\vec{\mu}$ over~$s$ such that $\vec{\mu}_d \ne 0$ and $\vec{\mu}_e = 0$ holds for all $e \not \in \Co(d)$.

\begin{example}
Any cut vector is a $\Co(d)$-pseudo-cut for some $d \in D$, by definition, and so are scalar multiples of cut vectors. The vector $\cvec{f}$ in Example~\ref{ex:cutvects}, however, is not a $\Co(d)$-pseudo-cut, since $\langle q_1,a\rangle \not \in \Co(\langle q_0, a\rangle)$ and $\langle q_0,a\rangle \not \in \Co(\langle q_1, a\rangle)$.
\end{example}

From a $\Co(d)$-pseudo-cut we can easily derive a $D$-normaliser:
\begin{lemma}\label{lem:colocal-pseudo-cut}
Let $\vec{\mu} \in \mathbb{R}^D$ be a $\Co(d)$-pseudo-cut.
Then $\frac{1}{\vec{\mu}_d}\vec{\mu}$ is a $D$-normaliser.
\end{lemma}
\begin{proof}
Let $w$ be an enabled word in~$M$ such that $d \then w$ is a cut containing~$d$.
Such a word exists (see the proof sketch of Proposition~\ref{prop:findcut}).
Since $(\cvec{d}^\top\Delta(w))^\top = \cvec{d \then w}$ is a $D$-normaliser (by Lemma~\ref{lem:cutnormalises}), it suffices to prove that $\frac{1}{\vec{\mu}_d}\vec{\mu}^\top\vec{z} = \cvec{d}^\top \Delta(w)\vec{z}$.

We can write $\vec{\mu}$ as $\sum_{d' \in \Co(d)} \vec{\mu}_{d'} \cvec{d'}$, so $\vec{\mu}^\top \Delta(w) = \sum_{d' \in \Co(d)} \vec{\mu}_{d'} \cvec{d'}^\top \Delta(w)$. For any $d' \in \Co(d) \setminus \{d\}$, let $w'$ be such that $\{d, d'\} \subseteq d \then w'$. Now we see that $d' \in d \then w w'$, and since $d \then w$ is a cut so are $d \then ww'$ and $d \then ww'w$. Thus,
\begin{equation*}
\cvec{d}^\top\Delta(w)\vec{z}  \ = \  \cvec{d}^\top\Delta(ww'w)\vec{z} \ \geq \ \cvec{d}^\top\Delta(w)\vec{z}+\cvec{d'}^\top\Delta(w)\vec{z},
\end{equation*}
which implies $\cvec{d'}^\top\Delta(w)\vec{z} = 0$ for every $d' \in \Co(d) \setminus \{d\}$. This means that
\begin{equation*}
\vec{\mu}^\top\Delta(w)\vec{z} = \sum_{d' \in \Co(d)}\vec{\mu}_{d'}\cvec{d'}^\top\Delta(w)\vec{z}  = \vec{\mu}_d\cvec{d}^\top\Delta(w)\vec{z}.
\end{equation*}
Since $\vec{\mu}$ is a pseudo-cut, this implies that $\frac{1}{\vec{\mu}_d}\vec{\mu}^\top\vec{z} = \frac{1}{\vec{\mu}_d}\vec{\mu}^\top\Delta(w)\vec{z} = \cvec{d}^\top\Delta(w)\vec{z}$.
\end{proof}
By Lemma~\ref{lem:colocal-pseudo-cut}, to find a $D$-normaliser it suffices to find a $\Co(d)$-pseudo-cut.
Fix a dominant eigenvector $\vec{y}$ of~$B_{D,D}$ so that $\vec{y}$ is strictly positive in all components.
One can compute such~$\vec{y}$ in time $O(|D|^\kappa)$.
By \cite[Lemma~8]{16BKKKMW-CAV} the vector $\vec{z}_{D}$ is also a dominant eigenvector of~$B_{D,D}$, hence $\vec{y}$ and $\vec{z}_{D}$ (the latter of which is yet unknown) are scalar multiples.
In order to compute a $\Co(d)$-pseudo-cut, we compute a basis for the space spanned by $\Delta(w)\vec{y}$ for all enabled words~$w$.
We use a technique similar to the one employed by Tzeng in \cite{tzen1992} for checking equivalence of probabilistic automata.
To make this more efficient, we compute separate basis vectors for each $s \in S$.
Define $\Delta'(t) \in \{0,1\}^{D \times D}$ as $\Delta'(t)_{\langle q_1,s_1\rangle,\langle q_2,s_2\rangle} = 1$ if $q_1=q_2$ and $s_1 = s_2 = t$ and 0 otherwise.
Note that $\Delta(s)\Delta'(s) = \Delta(s)$ holds for all $s \in S$.

\begin{restatable}{ourlemma}{lemcalcR}\label{lem:calcR}
Suppose $\vec{y} = B_{D,D} \vec{y}$ is given.
Denote by $V(s) \subseteq \mathbb{R}^D$ the vector space spanned by the vectors $\Delta'(s)\Delta(w)\vec{y}$ for $w \in S^*$ and $s \in S$.
Let $Q_{D, t} = (Q \times \{t\}) \cap D$ and let $E(t) = \{(s,t) \mid M_{s,t} > 0\}$ be the set of edges in~$M$ that end in~$t$.
One can compute a basis $R(s)$ of $V(s)$ for all $s \in S$ in time
{$O(|Q|^2 \sum_{t \in S} |Q_{D,t}||E(t)|)$}, where for each $\vec{r} \in R(s)$ we have $\vec{r} = \Delta'(s)\Delta(w)\vec{y}$ for some enabled word $s w$.
\end{restatable}
\begin{proof}[Proof sketch]
Fix an arbitrary total order $<_S$ on $S$.
We define a total order $\ll_S$ on $S^*$ as the ``shortlex'' order but with words read from right to left.
That is, the empty word~$\varepsilon$ is the smallest element, and for $v, w \in S^*$ and $s,t \in S$, we have $v s \ll_S w t$ if (1) $|v s| < |w t|$ or (2) $|vs| = |w t|$ and $s <_S t$ or (3) $|v s| = |w t|$ and $s = t$ and $v \ll_S w$.

We use a technique similar to the one by Tzeng in~\cite{tzen1992}.
At every step in the algorithm, $\unvisited$ is a set of pairs $(s w, \Delta'(s) \Delta(w) \vec{y})$.
We write $\min_{\ll_S}(\unvisited)$ to denote the pair in $\unvisited$ where $s w$ is minimal with respect to~$\mathord{\ll_S}$.
\begin{enumerate}
\item for each $s \in S$, let $R(s) := \{\Delta'(s)\vec{y}\}$ and $R(s)_\bot := \{\Delta'(s)\vec{y}\}$
\item $\unvisited := \{(s t,\Delta'(s)\Delta(t)\vec{y}) \mid M_{s,t} > 0\}$
\item while $\unvisited \neq \emptyset$: \\
\mbox{}\hspace{4mm} $(t w,\vec{u}) := \min_{\ll_S}(\unvisited)$; $\unvisited := \unvisited \setminus \{(t w,\vec{u})\}$ \\
\mbox{}\hspace{4mm} Using the Gram-Schmidt process\footnote{For good numerical stability, one should use the so-called Modified Gram-Schmidt process~\cite[Chapter~19]{Higham}.}, let $\vec{u}_\bot$ be the orthogonalisation of $\vec{u}$ against $R_\bot(t)$ \\
\mbox{}\hspace{4mm} if $\vec{u}_\bot \neq \vec{0}$, i.e., if $\vec{u}$ is linearly independent of $R_\bot(t)$: \\
\mbox{}\hspace{10mm}     $R(t) := R(t) \cup \{\vec{u}\}$ and $R_\bot(t) := R_\bot(t) \cup \{\vec{u}_\bot\}$ \\
\mbox{}\hspace{10mm}     $\unvisited := \unvisited \cup \{(s t w, \Delta'(s)\Delta(t)\vec{u}) \mid M_{s,t} > 0\}$
\item return $R(s)$ for all $s \in S$
\end{enumerate}
At any point and for all $s \in S$, the sets $R(s)$ and $R(s)_\bot$ span the same vector space, and this space is a subspace of~$V(s)$.
The sets $R(s)$ and $R(s)_\bot$ consist of linearly independent fibres over $s$, and these fibres are possibly nonzero only in the $Q_{D,s}$-components.
Hence $|\bigcup_{s \in S} R(s)| \leq |D|$ and thus there are at most $|D|$ iterations of the while loop that increase $\unvisited$.
At every iteration where $\vec{u}$ is dependent on $R(t)_\bot$ the set $\unvisited$ decreases by one, and therefore the algorithm terminates.
In~\cite[Appendix B.3]{kief2019} we prove that in the end we have that $R(s)$ spans~$V(s)$, and we analyse the runtime.
\end{proof}

\begin{example}\label{ex:spanex}
Let us return to our running example. We see that the vector $\vec{y} = (\vec{y}_{\langle q_0,a \rangle},\vec{y}_{\langle q_1,a \rangle},\vec{y}_{\langle q_1,b \rangle},\vec{y}_{\langle q_2,a \rangle},\vec{y}_{\langle q_2,b \rangle},\vec{y}_{\langle q_3,a \rangle})^\top = (2,1,3,1,3,2)^\top$ is a dominant eigenvector of $B_{D,D}$.
Fix the order $a <_S b$.
Step~1 initialises $R(a)$ to $\{\Delta'(a) \vec{y}\}$ and $R(b)$ to $\{\Delta'(b) \vec{y}\}$, where $\Delta'(a) \vec{y} = (2,1,0,1,0,2)^\top$ and $\Delta'(b) \vec{y} = (0,0,3,0,3,0)^\top$.
Step~2 computes $\Delta'(a)\Delta(a)\vec{y} = (1,2,0,2,0,1)^\top$, which is linearly independent of $\Delta'(a)\vec{y}$.
However, $\Delta'(b)\Delta(a)\vec{y} = (0,0,3,0,3,0)^\top = \Delta'(b)\vec{y}$.
Also, $\Delta'(a)\Delta(b)\vec{y} = (3,0,0,0,0,3)^\top = 2\Delta'(a)\vec{y}-\Delta'(a)\Delta(a)\vec{y}$ and $\Delta'(b)\Delta(b)\vec{y} = (0,0,3,0,3,0)^\top = \Delta'(b)\vec{y}$.
One can check that $\Delta'(a) \Delta(a a) \vec{y} = \Delta'(a) \vec{y}$ and $\Delta'(b) \Delta(a a) \vec{y} = \Delta'(b) \vec{y}$.
Hence the algorithm returns $R(a) = \{(2,1,0,1,0,2)^\top,(1,2,0,2,0,1)^\top\}$ and $R(b) = \{(0,0,3,0,3,0)^\top\}$.
\end{example}

Fix $d = \langle q,s\rangle \in D$ for the rest of the paper.
The following lemma characterises $\Co(d)$-pseudo-cuts in a way that is efficiently computable:
\begin{lemma}\label{lem:nullspclpc}
A vector $\vec{\mu} \in \Real^D$ with $\vec{\mu}_d = 1$ and $\vec{\mu}_e = 0$ for all $e \not \in \Co(d)$ is a $\Co(d)$-pseudo-cut if and only if $\vec{\mu}^\top \vec{r} = \vec{\mu}^\top \vec{y}$ holds for all $\vec{r} \in R(s)$.
\end{lemma}
For an intuition of the proof, consider the affine space, $\mathcal{F} \subseteq \Real^D$, affinely spanned by those $\Delta'(s) \Delta(w) \vec{y}$ for which $s w$ is enabled.
This affine space was alluded to in the beginning of this subsection and is visualised as a blue straight line on the right of Figure~\ref{fig:automaton}.
The shaded plane in this figure is the vector space of pseudo-cuts over~$s$.
This space is orthogonal to~$\mathcal{F}$.
The following lemma says that $\mathcal{F}$ is affinely spanned by the points in~$R(s)$.
This strengthens the property of~$R(s)$ in Lemma~\ref{lem:calcR} where $R(s)$ was defined to span a \emph{vector} space.
\begin{lemma} \label{lem:affine}
Let $w \in S^*$ be such that $s w$ is enabled.
By the definition of~$R(s)$ there are $\gamma_{\vec{r}} \in \Real^D$ for each $\vec{r} \in R(s)$ such that $\Delta'(s) \Delta(w) \vec{y} = \sum_{\vec{r} \in R(s)} \gamma_{\vec{r}} \vec{r}$.
We have $\sum_{\vec{r} \in R(s)} \gamma_{\vec{r}} = 1$.
\end{lemma}
\begin{proof}
Let $c$ be a cut containing~$d$.
Since $R(s)$ is a basis, for any $\vec{r} = \Delta'(s)\Delta(w_{\vec{r}})\vec{y} \in R(s)$ the word $s w_{\vec{r}}$ is enabled.
Therefore, $c \then w_{\vec{r}}$ is a cut and by Lemma~\ref{lem:cutnormalises} we have $\cvec{c \then w_{\vec{r}}}^\top \vec{y} = \cvec{c}^\top\vec{y}$.
Hence $\cvec{c}^\top \vec{r} = \cvec{c}^\top \Delta'(s) \Delta(w_{\vec{r}}) \vec{y} = \cvec{c}^\top \Delta(w_{\vec{r}}) \vec{y} = \cvec{c \then w_{\vec{r}}}^\top \vec{y} = \cvec{c}^\top\vec{y}$.
Moreover, we have:
\begin{align*}
\cvec{c}^\top\vec{y} & = \cvec{c}^\top\Delta(w)\vec{y} &&\text{since $s w$ is enabled and by Lemma~\ref{lem:cutnormalises}} & \\
& = \cvec{c}^\top \Delta'(s)\Delta(w)\vec{y}  &&\text{since $\cvec{c}$ is a fibre over $s$} \\
& = \cvec{c}^\top \sum_{\vec{r} \in R(s)} \gamma_{\vec{r}} \vec{r} &&\text{by the definition of $\gamma_{\vec{r}}$}\\
& = \cvec{c}^\top\vec{y} \sum_{\vec{r} \in R(s)}\gamma_{\vec{r}} &&\text{as argued above.}
\end{align*}
Therefore, $\sum_{\vec{r} \in R(s)}\gamma_{\vec{r}} = 1$.
\end{proof}
Now we can prove Lemma~\ref{lem:nullspclpc}:
\begin{proof}[Proof of Lemma~\ref{lem:nullspclpc}]
For the ``if'' direction, let $w$ be such that $sw$ is enabled, and it suffices to show that $\vec{\mu}^\top \Delta(w) \vec{y} = \vec{\mu}^\top \vec{y}$.
By Lemma~\ref{lem:affine} there are $\gamma_{\vec{r}}$ such that $\Delta'(s) \Delta(w) \vec{y} = \sum_{\vec{r} \in R(s)} \gamma_{\vec{r}} \vec{r}$ and $\sum_{\vec{r} \in R(s)} \gamma_{\vec{r}} = 1$.
We have:
\begin{equation*}
\vec{\mu}^\top \Delta(w) \vec{y} \ = \ \vec{\mu}^\top \Delta'(s)\Delta(w)\vec{y} \ = \ \sum_{\vec{r} \in R(s)} \gamma_{\vec{r}} \vec{\mu}^\top \vec{r} \ = \ \sum_{\vec{r} \in R(s)} \gamma_{\vec{r}} \vec{\mu}^\top \vec{y} \ = \ \vec{\mu}^\top \vec{y}\;,
\end{equation*}
where the last equality is from Lemma~\ref{lem:affine}.

For the ``only if'' direction, suppose $\vec{\mu}$ is a $\Co(d)$-pseudo-cut.
Let $\vec{r} = \Delta'(s)\Delta(w_{\vec{r}})\vec{y} \in R(s)$.
Then $s w_{\vec{r}}$ is enabled and $\vec{\mu}^\top \vec{r} = \vec{\mu}^\top\Delta'(s)\Delta(w_{\vec{r}})\vec{y} = \vec{\mu}^\top \Delta(w_{\vec{r}})\vec{y} = \vec{\mu}^\top \vec{y}$.
\end{proof}
\begin{example}
In Example~\ref{ex:cut} we derived that $\vec{y} = (2,1,3,1,3,2)^\top$ and $R(a) = \{(2,1,0,1,0,2)^\top,(1,2,0,2,0,1)^\top\}$.
The cut vector $\vec{\mu} = (1,0,0,1,0,0)^\top$ from Example~\ref{ex:cutvects} satisfies $\vec{\mu}^\top \vec{r} =  3 = \vec{\mu}^\top \vec{y}$ for both $\vec{r} \in R(a)$.
\end{example}

Using Lemmas \ref{lem:calccod}, \ref{lem:calcR} and~\ref{lem:nullspclpc} we obtain:
\begin{restatable}{ourproposition}{propfindcodcut}\label{prop:findcodcut}
Let $D \subseteq Q \times S$ be a recurrent SCC.
Denote by $T_D$ the set of edges of~$B_{D,D}$.
For $t \in S$, let~$E(t)$ denote the set of edges of~$M$ that end in $t$, and let $Q_{D,t} = (Q \times \{t\}) \cap D$.
Let $d = \langle q,s\rangle \in D$.
One can compute a $\Co(d)$-pseudo-cut in time 
$O(|D|^\kappa + |Q||D|+|\delta||T_D| + |Q|^2 \sum_{t \in S} |Q_{D,t}||E(t)|)$.
\end{restatable}
Now our main result follows, which we restate here:
\thmPMCMCUBA*

\section{Discussion} \label{sec-discussion}

We have analysed two algorithms for computing normalisers: the cut-based one by Baier et al.~\cite{16BKKKMW-CAV,baieunp1}, and a new one, which draws from techniques by Protasov and Voynov~\cite{Protasov17} for the analysis of matrix semigroups.
The first approach is purely combinatorial, and in terms of the automaton, an efficient implementation runs in time $O(|Q|^3 |\delta| + |\delta|^2) = O(|Q|^3 |\delta|)$ (Proposition~\ref{prop:findcut}).

The second approach combines a linear-algebra component to compute~$R(s)$ with a combinatorial algorithm to compute the co-reachability set~$\Co(d)$.
In terms of the automaton, the linear-algebra component runs in time $O(|Q|^3)$ (Lemma~\ref{lem:calcR}), while the combinatorial part runs in time $O(|\delta|^2)$, leading to an overall runtime of $O(|Q|^3 + |\delta|^2)$.
Note that for all $r \in [1,2]$, if $|\delta| = \Theta(|Q|^r)$ then the second approach is faster by at least a factor of~$|Q|$.

Although it is not the main focus of this paper, we have analysed also the model-checking problem, where a non-trivial Markov chain is part of the input.
The purely combinatorial algorithm runs in time $O(|Q|^\kappa |S|^\kappa + |Q|^3 |\delta| |S|+|\delta|^2 |E|)$, and the linear-algebra based algorithm in time $O(|Q|^\kappa |S|^\kappa + |Q|^3 |E|+|\delta|^2 |E|)$.
There are cases in which the latter is asymptotically worse, but not if $\kappa = 3$ (i.e., solving linear systems in a normal way such as Gaussian elimination) or if $|E|$ is $O(|S|)$.

It is perhaps unsurprising that a factor of $|\delta|^2$ from the computation of~$\Co(d)$ occurs in the runtime, as it also occurs when one merely verifies the unambiguousness of the automaton, by searching the product of the automaton with itself.
Can the factor $|\delta|^2$ (which may be quartic in~$|Q|$) be avoided?


\bibliographystyle{plain}
\bibliography{ms}

\begin{thebibliography}{10}

\bibitem{baie2008}
Christel Baier and Joost-Pieter Katoen.
\newblock {\em Principles of model checking}.
\newblock MIT press, 2008.

\bibitem{16BKKKMW-CAV}
Christel Baier, Stefan Kiefer, Joachim Klein, Sascha Kl{\"{u}}ppelholz, David
  M{\"{u}}ller, and James Worrell.
\newblock {M}arkov chains and unambiguous {B}{\"{u}}chi automata.
\newblock In {\em Proceedings of Computer Aided Verification (CAV)}, volume
  9779 of {\em LNCS}, pages 23--42, 2016.

\bibitem{baieunp1}
Christel Baier, Stefan Kiefer, Joachim Klein, Sascha Kl{\"{u}}ppelholz, David
  M{\"{u}}ller, and James Worrell.
\newblock {M}arkov chains and unambiguous automata.
\newblock Draft journal submission. Available at
  https://arxiv.org/abs/1605.00950, 2019.

\bibitem{book:bermanP94}
Abraham Berman and Robert~J. Plemmons.
\newblock {\em Nonnegative matrices in the mathematical sciences}.
\newblock SIAM, 1994.

\bibitem{Bodirsky04}
Manuel Bodirsky, Tobias G{\"a}rtner, Timo von Oertzen, and Jan Schwinghammer.
\newblock Efficiently computing the density of regular languages.
\newblock In {\em LATIN 2004: Theoretical Informatics}, pages 262--270.
  Springer, 2004.

\bibitem{bunc1974}
James~R. Bunch and John~E. Hopcroft.
\newblock Triangular factorization and inversion by fast matrix multiplication.
\newblock {\em Mathematics of Computation}, 28:231--236, 1974.

\bibitem{BusRubVar04}
Doron Bustan, Sasha Rubin, and Moshe~Y. Vardi.
\newblock Verifying $\omega$-regular properties of {Markov} chains.
\newblock In {\em 16th International Conference on Computer Aided Verification
  (CAV)}, volume 3114 of {\em Lecture Notes in Computer Science}, pages
  189--201. Springer, 2004.

\bibitem{CY95}
Costas Courcoubetis and Mihalis Yannakakis.
\newblock The complexity of probabilistic verification.
\newblock {\em Journal of the ACM}, 42(4):857--907, 1995.

\bibitem{GraedelThomasWilke02}
Erich Gr{\"a}del, Wolfgang Thomas, and Thomas Wilke, editors.
\newblock {\em Automata, Logics, and Infinite Games: A Guide to Current
  Research}, volume 2500 of {\em Lecture Notes in Computer Science}. Springer,
  2002.

\bibitem{Higham}
Nicholas~J. Higham.
\newblock {\em Accuracy and Stability of Numerical Algorithms}.
\newblock SIAM, second edition, 2002.

\bibitem{jame1978}
M.~James.
\newblock The generalised inverse.
\newblock {\em The Mathematical Gazette}, 62:109--114, 1978.

\bibitem{kief2019}
Stefan Kiefer and Cas Widdershoven.
\newblock Efficient analysis of unambiguous automata using matrix semigroup
  techniques (full version).
\newblock 2016.
\newblock \url{http://arxiv.org/abs/1906.10093}.

\bibitem{Kulkarni}
Vidyadhar~G. Kulkarni.
\newblock {\em Modeling and Analysis of Stochastic Systems}.
\newblock Chapman \& Hall, 1995.

\bibitem{gall2014}
Fran\c{c}ois Le~Gall.
\newblock Powers of tensors and fast matrix multiplication.
\newblock In {\em Proceedings of the 39th International Symposium on Symbolic
  and Algebraic Computation}, ISSAC'14, pages 296--303. ACM, 2014.

\bibitem{petk2009}
Marko~D. Petković and Predrag~S. Stanimirović.
\newblock Generalised matrix inversion is not harder than matrix
  multiplication.
\newblock {\em Journal of Computational and Applied Mathematics}, 230:270--282,
  2009.

\bibitem{Protasov17}
V.Yu. Protasov and A.S. Voynov.
\newblock Matrix semigroups with constant spectral radius.
\newblock {\em Linear Algebra and its Applications}, 513:376--408, 2017.

\bibitem{tarj1972}
Robert Tarjan.
\newblock Depth-first search and linear graph algorithms.
\newblock {\em SIAM journal on computing}, 1(2):146--160, 1972.

\bibitem{tzen1992}
Wen-Guey Tzeng.
\newblock A polynomial-time algorithm for the equivalence of probabilistic
  automata.
\newblock {\em SIAM Journal on Computing}, 21(2):216--227, 1992.

\end{thebibliography}

\newpage
\appendix
\section{The Number of Transitions Can Be Quadratic.}
\label{app-A}
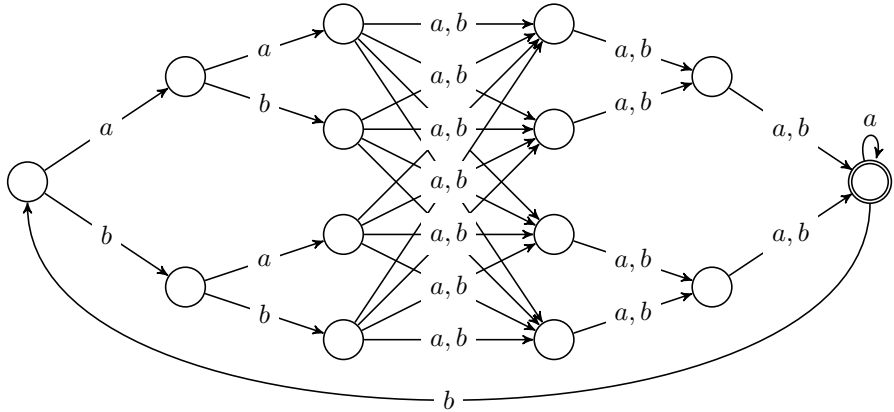
\begin{figure}[!h]
\begin{tikzpicture}[-{>[scale=5]},color=black,scale=0.7,semithick,,>=stealth']
	\node[circle,draw=black,minimum size=15pt] (1) at (0,0) {};
	\node[circle,draw=black,minimum size=15pt] (2) at (3,2) {};
	\node[circle,draw=black,minimum size=15pt] (3) at (3,-2) {};
	\node[circle,draw=black,minimum size=15pt] (4) at (6,3) {};
	\node[circle,draw=black,minimum size=15pt] (5) at (6,1) {};
	\node[circle,draw=black,minimum size=15pt] (6) at (6,-1) {};
	\node[circle,draw=black,minimum size=15pt] (7) at (6,-3) {};
	\node[circle,draw=black,minimum size=15pt] (8) at (10,3) {};
	\node[circle,draw=black,minimum size=15pt] (9) at (10,1) {};
	\node[circle,draw=black,minimum size=15pt] (10) at (10,-1) {};
	\node[circle,draw=black,minimum size=15pt] (11) at (10,-3) {};
	\node[circle,draw=black,minimum size=15pt] (12) at (13,2) {};
	\node[circle,draw=black,minimum size=15pt] (13) at (13,-2) {};
	\node[circle,draw=black,minimum size=15pt,accepting] (14) at (16,0) {};
	
	\draw (1)--(2) node[midway,fill=white] {$a$};
	\draw (1)--(3) node[midway,fill=white] {$b$};
	\draw (2)--(4) node[midway,fill=white] {$a$};
	\draw (2)--(5) node[midway,fill=white] {$b$};
	\draw (3)--(6) node[midway,fill=white] {$a$};
	\draw (3)--(7) node[midway,fill=white] {$b$};
	\draw (4)--(8) node[midway,fill=white] {$a,b$};
	\draw (4)--(9) node[midway,fill=white] {$a,b$};
	\draw (4)--(10) node[midway,fill=white] {$a,b$};
	\draw (4)--(11) node[midway,fill=white] {$a,b$};
	\draw (5)--(8) node[midway,fill=white] {$a,b$};
	\draw (5)--(9) node[midway,fill=white] {$a,b$};
	\draw (5)--(10) node[midway,fill=white] {$a,b$};
	\draw (5)--(11) node[midway,fill=white] {$a,b$};
	\draw (6)--(8) node[midway,fill=white] {$a,b$};
	\draw (6)--(9) node[midway,fill=white] {$a,b$};
	\draw (6)--(10) node[midway,fill=white] {$a,b$};
	\draw (6)--(11) node[midway,fill=white] {$a,b$};
	\draw (7)--(8) node[midway,fill=white] {$a,b$};
	\draw (7)--(9) node[midway,fill=white] {$a,b$};
	\draw (7)--(10) node[midway,fill=white] {$a,b$};
	\draw (7)--(11) node[midway,fill=white] {$a,b$};
	\draw (8)--(12) node[midway,fill=white] {$a$};
	\draw (9)--(12) node[midway,fill=white] {$b$};
	\draw (10)--(13) node[midway,fill=white] {$a$};
	\draw (11)--(13) node[midway,fill=white] {$b$};
	\draw (12)--(14) node[midway,fill=white] {$a$};
	\draw (13)--(14) node[midway,fill=white] {$b$};
	
	
	\draw (14) edge[out=270,in=270,bend angle=180,looseness=0.8] node[midway,fill=white] {$a,b$} (1);
\end{tikzpicture}
\caption{An unambiguous automaton with a single recurrent strongly connected SCC and a quadratic number of edges.}\label{fig:manyarrows}
\end{figure}

The example in Figure \ref{fig:manyarrows} shows a strongly connected UBA where the number of transitions is quadratic in $|Q|$ for a fixed alphabet $\Sigma = \{a,b\}$.
Indeed, for any $n \geq 2$ we can create a similar strongly connected UBA with $|Q| = O(2^{n})$ states and at least $(2^n)^2$ transitions.

Such a UBA can be created as follows: take two (directed) complete binary trees $T_1$ and~$T_2$ where each node except for the leaves has $2$ outgoing edges labelled by $a,b$, respectively.
In~$T_2$, flip the directions of all edges. 
From the former root of~$T_2$, add a transition, labelled by both $a$ and~$b$, to the root of~$T_1$.
From each leaf of $T_1$ add transitions to each former leaf of~$T_2$, all labelled by both $a$ and~$b$.
The resulting graph is a strongly connected UBA with a number of arrows that is quadratic in the number of states, similar to the one in Figure \ref{fig:manyarrows}.

\section{Missing Proofs}

\subsection{Proof From Section~\ref{sec:prelim}} \label{app:prelims}

We prove Lemma~\ref{lem:diamondfree} from the main body:
\lemdiamondfree*
\begin{proof}
Let $\cA = (Q,\Sigma,\delta,Q_0,F)$ be the given UBA.
First, we remove every state $q \in Q$ that is not reachable from $Q_0$, together with all its incoming and outgoing edges (i.e., transitions with their labels).
Breadth-first search starting from every $q_0 \in Q_0$ costs at most $O(|Q|^2 + |Q||\delta||\Sigma|)$ time overall, as $|\delta||\Sigma|$ is at least the number of edges in $\cA$.
Now we need to remove all the diamonds in the remaining automaton.

If there exists a diamond from $q$ to $s$, then, since $q$ is reachable, we see that $\cL(\cA_s) = \emptyset$ holds by unambiguousness, so we can remove $s$ (and all its incoming and outgoing edges) from the automaton without changing the language.
Consider the product automaton $\cA \times \cA$ with states $(Q \times Q) \cup \{r\}$ and alphabet $\Sigma \cup \{\$\}$ for a fresh letter~$\$$.
There exists a transition between $(q,q')$ and $(s,s')$ labelled by $a \in \Sigma$ if and only if there exist an $a$-transition between $q$ and~$s$ and an $a$-transition between $q'$ and~$s'$.
There also exists an $\$$-transition between $r$ and $(q,q)$ for every $q \in Q_0$.
The number of transitions in $\cA \times \cA$ is at most $|\delta|^2|\Sigma| + |Q|$, as any transition is either an $\$$-transition from $r$ to $q \in Q_0$, or it is represented by some $a \in \Sigma$ and pairs $(q,q')$ and $(s,s')$ such that there exist $a$-transitions from $q$ to~$s$ and from $q'$ to~$s'$.

Finally we perform a breadth-first search on $\cA \times \cA$, starting in~$r$, and remove $q$ from~$\cA$ every time we encounter an edge between $(s, s')$ and $(q,q)$ with $s \neq s'$.
Since we had removed all unreachable states in~$\cA$, this means that for any diamond from $q$ to $t$, the state~$(q,q)$ is reachable from~$r$, and traversing the diamond means that at some point we will encounter an edge from $(s,s')$ with $s \neq s'$ to~$(t',t')$ for some state~$t'$.
Removing~$t'$ eliminates the diamond without altering the language of~$\cA$.
Breadth-first search on $\cA \times \cA$ costs time $O(|\delta|^2|\Sigma|)$, concluding the proof.
\end{proof}

\subsection{Proofs From Section~\ref{sub:cuts}} \label{app:cuts}

We prove Lemma~\ref{lem:calccod} from the main body:
\lemcalccod*
\begin{proof}
Let $P = (V,A)$ denote the graph with vertices in $V = \{(p,p',s) \in Q \times Q \times S \mid (p,s), (p',s) \in D\}$ and an edge from $(p,p',s)$ to $(r,r',t)$ whenever there are edges from $\langle p,s \rangle$ to~$\langle r,t \rangle$ and from $\langle p',s \rangle$ to~$\langle r',t \rangle$ in~$B_{D,D}$.
We have $|V| \le |Q| |D|$ and $|A| \le |\delta| |T|$.
Define $\pi_S : V^* \rightarrow S^*$ as the function mapping $v_0 v_1\ldots v_n$ to $s_0 s_1\ldots s_n$, where $v_i = (p_i,p'_i,s_i)$.
Let $v_0 v_1 \ldots v_n$ be a path from $(q,q,s)$ to $(q,q',s)$ in $P$, $v_0 = (q,q,s)$ and $v_n = (q,q',s)$. Then $w = \pi_S(v_1\ldots v_n)$ is a sequence of states in $S$ such that $\{\langle q,s\rangle, \langle q',s\rangle\} \subseteq \langle q,s\rangle \then w$.
We can (breadth-first) search~$P$ to calculate $\Co(d)$ in $O(|V|+|A|)=O(|Q| |D| + |\delta| |T|)$.
Using the breadth-first search tree, we can save the shortest paths from $(q,q,s)$ to $(q,q',s)$ for each $e = \langle q', s\rangle \in \Co(d)$.
Since $\Co(d)$ is a fibre, we have $|\Co(d)| \le |Q|$.
Hence we can compute in time $O(|Q| |D| + |\delta| |T| + |Q| |V|) \le O(|Q|^2 |D| + |\delta| |T|)$ the list~$\CoPath(d)$ with $|\CoPath(d)(e)| \leq |Q| |D|$.
\end{proof}

We prove Proposition~\ref{prop:findcut} from the main body:
\propfindcut*
\begin{proof}
We continue from the proof sketch in the main body.
Step~1 of the algorithm computes $\Co(d)$ and $\CoPath(d)$ in time $O(|Q|^2 |D| + |\delta| |T|)$ using Lemma \ref{lem:calccod}.
We have $|\CoPath(d)(e)| \le |Q| |D|$, so the (final) word~$w$ has length at most $|Q|^2 |D|$.
For each state $v_{i-1}$ in~$w$ one can update the set~$\Survives$ in the inner loop in time $O(|\delta|)$
by going through the automaton transitions labelled by~$v_{i-1}$.
Hence step~3 of the algorithm takes time $O(|Q|^2 |\delta| |D|)$.
Step~4 is to calculate $d \then w$ by essentially the same computation as in the inner loop, but forwards instead of backwards and with fixed~$w$, hence also in time $O(|Q|^2 |\delta| |D|)$.
The total runtime is $O(|Q|^2 |\delta| |D| + |\delta| |T|)$.
\end{proof}

We prove Theorem~\ref{thm:PMC-MC-UBA-cut} from the main body:

\thmPMCMCUBAcut*
\begin{proof}
There are at most $|\delta| |E|$ edges in~$B$.
Thus, summing the quantities $|D|$ and~$|T|$ from Proposition~\ref{prop:findcut} over all SCCs~$D$ gives at most $|Q| |S|$ and $|\delta| |E|$, respectively.
Hence, by Proposition~\ref{prop:findcut}, one can compute normalisers for all recurrent SCCs in time $O(|Q|^3 |\delta| |S| + |\delta|^2 |E|)$.
Using Proposition~\ref{prop:mccomp} we get a total runtime of $O(|Q|^\kappa |S|^\kappa + |Q|^3 |\delta| |S|+|\delta|^2 |E|)$.
\end{proof}

\subsection{Proofs From Section~\ref{sub:pseudo-cuts}} \label{app:pseudocuts}

The following lemma is used in the proof of Lemma~\ref{lem:calcR}:
\begin{restatable}{ourlemma}{lemorderedlindep}\label{lem:orderedlindep}
Let $t \in S$ and $x \in S^*$. 
Let $U \subseteq S^*$ be a set of words with $u \ll_S x$ for all $u \in U$ and $\Delta'(t)\Delta(x)\vec{y} = \sum_{u \in U} \gamma_u \Delta'(t)\Delta(u)\vec{y}$ for some $\gamma_u \in \Real$,  $u \in U$.
Let $s \in S$ and $w \in S^*$. 
Then there exists a set $U'$ of words $u' \ll_S w t x$ such that $\Delta'(s)\Delta(w t x)\vec{y} = \sum_{u' \in U'} \gamma_{u'} \Delta'(s)\Delta(u')\vec{y}$ holds for some $\gamma_{u'} \in \Real$, $u' \in U'$.
\end{restatable}
\begin{proof}
We have:
\begin{align*}
\Delta'(s)\Delta(w t x)\vec{y}
& \ = \ \Delta'(s) \Delta(w) \Delta(t) \Delta(x)\vec{y}  \\
& \ = \ \Delta'(s) \Delta(w) \Delta(t) \Delta'(t) \Delta(x)\vec{y} \\
& \ = \ \Delta'(s) \Delta(w) \Delta(t)\Big(\sum_{u \in U} \gamma_u \Delta'(t)\Delta(u)\vec{y}\Big) \\
& \ = \ \sum_{u \in U}\gamma_u\Delta'(s)\Delta(w t u)\vec{y}\;.
\end{align*}
Hence choose $U' = \{w t u \mid u \in U\}$.
\end{proof}

Now we can prove Lemma~\ref{lem:calcR} from the main body:
\lemcalcR*
\begin{proof}
We use induction to prove that for each~$s$ the span of~$R(s)$ contains~$V(s)$.
Denote all nonempty words in $S^*$ by $v^{(1)},v^{(2)},\ldots$ such that $v^{(i)} \ll_S v^{(j)}$ if and only if $i < j$.
For each~$v^{(i)}$, let $s^{(i)} \in S$ and $w^{(i)} \in S^*$ be such that $v^{(i)} = s^{(i)} w^{(i)}$.
For each $v^{(i)}$ such that $w^{(i)}$ is the empty word, $\Delta'(s^{(i)})\vec{y}$ is linearly dependent on $R(s^{(i)})$.
Our induction hypothesis is that for each $v^{(i)}$ for $i \leq n$, we have that $\Delta'(s^{(i)})\Delta(w^{(i)})\vec{y}$ is dependent on $R(s^{(i)})$.
For the inductive step, consider $v^{(n+1)}$ and distinguish two cases:
\begin{itemize}
	\item  Either $\Delta'(s^{(n+1)})\Delta(w^{(n+1)})\vec{y}$ was visited by the algorithm: \\
 Then either $\Delta'(s^{(n+1)})\Delta(w^{(n+1)})\vec{y}$ was shown to be dependent on $R(s^{(n+1)})$ at that time, or $\Delta'(s^{(n+1)})\Delta(w^{(n+1)})\vec{y}$ was independent of $R(s^{(n+1)})$ at that time, after which $\Delta'(s^{(n+1)})\Delta(w^{(n+1)})\vec{y}$ was added to $R(s^{(n+1)})$.
	\item Or $\Delta'(s^{(n+1)})\Delta(w^{(n+1)})\vec{y}$ was not visited by the algorithm: \\
 Then for some suffix $v^{(q)}$ of $v^{(n+1)}$ the vector $\Delta'(s^{(q)})\Delta(w^{(q)})\vec{y}$ was visited by the algorithm and was found to be dependent on $R(s^{(q)})$ at the time, which was dependent on vectors of the form $\Delta'(s^{(q)})\Delta(w^{(p)})\vec{y}$, with $w^{(p)} \ll_S w^{(q)}$.
 By Lemma \ref{lem:orderedlindep} the vector $\Delta'(s^{(n+1)})\Delta(w^{(n+1)})\vec{y}$ is dependent on vectors of the form $\Delta'(s^{(n+1)})\Delta(w^{(m)})\vec{y}$ with $w^{(m)} \ll_S w^{(n+1)}$.
 Since $s^{(n+1)}w^{(m)} \ll_S v^{(n+1)}$, by the induction hypothesis all of these vectors are dependent on $R(s^{(n+1)})$. Therefore, $\Delta'(s^{(n+1)})\Delta(w^{(n+1)})\vec{y}$ is dependent on $R(s^{(n+1)})$.
\end{itemize}

Concerning runtime, for every vector $\vec{u} = \Delta'(t)\Delta(w)\vec{y}$ visited by the algorithm, $\vec{u}$ is orthogonalised against~$R(t)$.
This costs $O(|Q|^2)$ time, as $\vec{u}$ and every vector in~$R(t)$ is a fibre over~$t$.
Moreover, one can calculate~$\Delta'(s) \Delta(t) \vec{u}$ from~$\vec{u}$ in time $O(|Q|^2)$, as $\Delta'(s) \Delta(t)$ is possibly nonzero only in a $|Q| \times |Q|$ submatrix.
The algorithm adds $|E(t)|$ vectors to $\unvisited$ each time a vector gets added to $R(t)$, and the latter happens at most $|Q_{D,t}|$ times since the nonzero elements only occur in $\langle q, t\rangle$ with $\langle q, t\rangle \in D$.
This gives us the bound of $O(|Q|^2 \sum_{t \in S} |Q_{D,t}||E(t)|)$.
\end{proof}

We prove Proposition~\ref{prop:findcodcut} from the main body:

\propfindcodcut*

\begin{proof}
One can compute~$\vec{y}$ in time $O(|D|^\kappa)$.
By Lemma \ref{lem:calccod}, one can calculate~$\Co(d)$ in time 
{$O(|Q||D|+|\delta||T_D|)$}.
By Lemma~\ref{lem:calcR}, one can compute~$R(s)$ in time $O(|Q|^2 \sum_{t \in S}|Q_{D,t}| |E(t)|)$.
The $O(|D|)$ equations 
from Lemma~\ref{lem:nullspclpc} can be solved in time $O(|D|^\kappa)$.
The total runtime is $O(|D|^\kappa + |Q||D|+|\delta||T_D| + |Q|^2 \sum_{t \in S}|Q_{D,t}| |E(t)|)$.
\end{proof}

We prove Theorem~\ref{thm:PMC-MC-UBA} from the main body:

\thmPMCMCUBA*
\begin{proof}
Denote by $\mathcal{D}$ the set of accepting recurrent SCCs.
Using Proposition~\ref{prop:findcodcut} we compute a normaliser for each of them.
Since there are at most $|\delta||E|$ edges in $B$, we have $\sum_{D \in \mathcal{D}} |T_D| \le |\delta| |E|$.
Hence, $\sum_{D \in \mathcal{D}} |D|^\kappa + |Q| |D| + |\delta| |T_D|$ is $O(|Q|^\kappa |S|^\kappa + \delta^2 |E|)$.
For any $t \in S$ and different SCCs $D,D' \in \mathcal{D}$, the sets $Q_{D,t}$ and $Q_{D',t}$ are disjoint.
Thus,
\[
\sum_{t \in S} \sum_{D \in \mathcal{D}} |Q_{D,t}| |E(t)| \ \le \ \sum_{t \in S} |Q| |E(t)| \ = \ |Q| |E|\,.
\]
Hence by Proposition~\ref{prop:findcodcut}, one can compute normalisers for all accepting recurrent SCCs in time $O(|Q|^\kappa |S|^\kappa + |Q|^3 |E|+|\delta|^2 |E|)$.
Using Proposition~\ref{prop:mccomp} we get a total runtime of $O(|Q|^\kappa |S|^\kappa + |Q|^3 |E|+|\delta|^2 |E|)$.
\end{proof}

\end{document}